%% file: main.tex
\documentclass[11pt,letterpaper]{article}
\usepackage{maxcut}
\usepackage{psmacros}
\usepackage{sos-macros}
\usepackage{tablefootnote}

\newcommand{\Snote}[1]{\footnote{\color{Cerulean}Sam: #1}}

\title{Faster MAX-CUT on Bounded Threshold Rank Graphs}
\author{
Prashanti Anderson\thanks{MIT EECS and CSAIL. Email: \texttt{paanders@mit.edu}. Supported by NSF award no. 2238080.}
\and
Samuel B. Hopkins\thanks{MIT EECS and CSAIL. Email: \texttt{samhop@mit.edu}. Supported by NSF award no. 2238080, MLA@CSAIL, FinTechAI@CSAIL, MIT Research Support Committee, MIT-Google Program for Computing Innovation.}
\and
Amit Rajaraman\thanks{MIT EECS and CSAIL. Email: \texttt{amit\_r@mit.edu}. Supported by a MathWorks Fellowship.}
\and
David Steurer\thanks{ETH Zurich. Email: \texttt{dsteurer@inf.ethz.ch}.}
}
\date{\today}

\begin{document}
\maketitle

\begin{abstract}
  We design new algorithms for approximating 2CSPs on graphs with bounded threshold rank, that is, whose normalized adjacency matrix has few eigenvalues larger than $\eps$, smaller than $-\eps$, or both.
  Unlike on worst-case graphs, 2CSPs on bounded threshold rank graphs can be $(1+O(\eps))$-approximated efficiently.
  Prior approximation algorithms for this problem run in time exponential in the threshold rank and $1/\eps$.
  Our algorithm has running time which is \emph{polynomial} in $1/\eps$ and exponential in the threshold rank of the \emph{label-extended} graph, and near-linear in the input size.
  As a consequence, we obtain the first $(1+O(\eps))$ approximation for MAX-CUT on bounded threshold rank graphs running in $\poly(1/\eps)$ time.
  We also improve the state-of-the-art running time for 2CSPs on bounded threshold rank graphs from polynomial in input size to near-linear via a new comparison inequality between the threshold rank of the label-extended graph and base graph.
  Our algorithm is a simple yet novel combination of subspace enumeration and semidefinite programming.
\end{abstract}

\thispagestyle{empty}
\setcounter{page}{0}
\newpage

% \tableofcontents

% \thispagestyle{empty}
% \setcounter{page}{0}

% \newpage

\input{content/intro}

\input{content/prelims}

% \input{maxcut-alg}

\input{content/large-alphabet}

\input{content/rank-bound}

\clearpage

\addcontentsline{toc}{section}{References}
\bibliographystyle{alpha}
\bibliography{refs}

\end{document}

%% file: content/intro.tex
%!TEX root = ./main.tex

\section{Introduction}

Constraint satisfaction problems (CSPs) lie at the heart of algorithms, discrete optimization, computational complexity, and more.
Classic examples like 3-SAT and MAX-CUT appeared in Karp's original list of 21 NP-complete problems; the study of both algorithms and computational complexity of CSPs for the last 50 years has been inextricably linked to breakthroughs like the PCP theorem, constructions of error-correcting codes, algorithm design techniques like linear and semidefinite programming, the dichotomy theorem, and the unique games theory, to name a few.

One fruitful avenue for algorithm development in the face of NP-hardness of general CSPs has been the design of algorithms which perform well on CSPs whose constraints form a graph or hypergraph with some beyond-worst-case structure.
Early examples in this spirit are dynamic programming algorithms for CSPs on bounded tree-width graphs \cite{kumar1992algorithms}, message-passing algorithms inspired by statistical physics for random CSPs \cite{braunstein2005survey}, and greedy algorithms for dense CSPs \cite{arora1995polynomial,MS08}.

A more recent line of work focuses on CSPs whose constraint graphs exhibit spectral structure---one or two-sided expansion, or more generally, having low \emph{threshold rank}, that is, having few large eigenvalues.
This generalizes the dense graph setting, and so far led to the subexponential algorithm for unique games \cite{ABS15}, efficient algorithms for (list-)decoding high-rate error-correcting codes \cite{uniquedecodingtashma, 10.1145/3406325.3451126}, new approximate counting and sampling algorithms \cite{pmlr-v49-risteski16, 10.1145/3313276.3316299}, new regularity lemmas \cite{OGT13}, and more generally the insight that CSPs are computationally tractable when their constraint graphs exhibit ``local to global'' phenomena \cite{BRS11}.
These works focus on 2CSPs to make things conceptually simple; \cite{ajt19} shows that insights from 2CSPs transfer to higher-arity CSPs.
The high-level questions addressed by these works and ours are:
\begin{quote}
    \centering
    \emph{For which MAX-2CSP instances can efficient algorithms find close-to-optimal assignments? What are the fastest algorithms to do so?}
\end{quote}

To go further we need some notation.
To an instance $\phi$ of an $n$-variable 2CSP we associate an $n$-vertex undirected graph $G$ with an edge $i \sim j$ in $G$ iff variables $i$ and $j$ participate in a constraint together in $\phi$.
Let $A$ be the normalized adjacency matrix of $G$, with eigenvalues in $[-1,1]$.
The $\epsilon$ threshold rank of $A$, denoted $\rank_{\geq \eps}(A)$, is the number of eigenvalues of value at least $\eps$.
We discuss some variants of threshold rank below.

\paragraph{Our Contributions.}
Our main contribution is a very simple but powerful approach to minimization of quadratic forms $x^\top A x$ over the Boolean hypercube, solid hypercube, and more generally over arbitrary convex subsets of $\R^n$ admitting an efficient separation oracle, whose running time is parameterized by the number of negative eigenvalues of $A$.
Note that if $A \succeq 0$, then minimization of $x^\top A x$ is a convex problem which can be solved in polynomial time, so running time parameterized by the number of negative eigenvalues of $A$ is sensible.

We apply this approach to obtain two main algorithmic results for 2CSPs on low threshold rank constraint graphs.
The first result concerns perhaps the simplest 2CSP: MAX-CUT.
MAX-CUT on low threshold rank graphs was studied in \cite{BRS11,OGT13}, and the special case of dense graphs was extensively studied in \cite{arora1995polynomial,fernandez1996max,yoshida2014approximation,FLP15,MS08}.

Our first main result is a $(1+O(\eps))$ approximation algorithm for MAX-CUT on $n$-vertex, $m$-edge graphs which runs in time $(1/\eps)^{O(\rank_{\leq -\eps})} \cdot \tilde{O}(m+n)$, where $\rank_{\leq -\eps}$ is the number of eigenvalues of the normalized adjacency matrix less than $-\eps$. 
No prior work for MAX-CUT on graphs satisfying any eigenvalue-based condition obtains a $(1+O(\eps))$ approximation in $\poly(1/\eps)$ time; prior algorithms always incur an exponential dependence on $1/\eps$.

Our second main result is a new $(1+O(\eps))$ approximation algorithm for general 2CSPs with finite alphabet $[q]$, where $q \in \N$, with running time $(1/\eps)^{O(\rank_{\geq \eps^2}/\eps^4)} \cdot \widetilde{O}(m + n)$, where $n$ is the number of variables and $m$ is the number of constraints/clauses.
The previous state-of-the-art algorithm by Barak--Raghavendra--Steurer \cite{BRS11} obtains the same guarantee only in time $\exp(\widetilde{O}(\rank_{\geq \eps^2}/\eps^4)) \cdot (mn)^{O(1)}$, relying on the sum-of-squares hierarchy.
Thus, our algorithm is the first for general 2CSPs on bounded threshold rank graphs to obtain near-linear running time; indeed, this is true even in the setting of expanders, where $\rank_{\geq \eps^2}=1$.\footnote{We discuss below the relationship between our work and algorithms for restricted classes of 2CSPs based on regularity lemmas. In short, while near-linear-time regularity decompositions exist in prior work and give a plausible alternative route to a near-linear time algorithm for general 2CSPs on low threshold rank graphs, no prior work carries this out.}
We discuss below some interpretation of the parameter $\rank_{\geq \eps^2}$, via higher-order Cheeger inequalities \cite{10.1145/2665063}.

\paragraph{Our Techniques, in Brief}
Our algorithms combine two key ingredients: subspace enumeration and semidefinite programming (SDP).
Both have been used extensively in algorithms for CSPs, but the way we combine them is novel.
Prior subspace enumeration algorithms were limited to nearly-satisfiable instances,
and prior SDP-based algorithms (a) did not achieve $\poly(1/\eps)$ running time for MAX-CUT and (b) did not achieve near-linear running time for $(1+O(\eps))$ approximation with respect to $m$ and $n$.
Unlike prior algorithms for CSPs based on SDPs, ours solves a \emph{series} of very simple SDPs, one for each element of a discretization of a subspace constructed from the underlying CSP.

We view the simplicity of our algorithm and analysis as a contribution in itself.
The entire algorithm and analysis for a variant running in $n^{O(1)}$ time rather than $\widetilde{O}(n+m)$ time spans only a few pages, and the running time can be improved to $\widetilde{O}(n+m)$ by appeal to standard matrix multiplicative weights techniques \cite{fastsdp2csps10}.
We manage to avoid the use of higher levels of convex hierarchies such as sum-of-squares.

Our analysis relies on a new inequality relating the threshold rank of $G$ to the threshold rank of the so-called \emph{label-extended} graph of a CSP instance $\phi$.
We defer the formal definition, but in brief, the label-extended graph for a CSP with $n$ variables and alphabet size $q$ is a graph with vertex set $[n] \times [q]$, where vertex $(i,a)$ represents ``variable $i$ is set to value $a$''.
Prior works relating the threshold rank of the label-extended graph to the threshold rank of $G$ either apply only to restricted classes of CSPs or lose polynomial factors in $n$; ours incurs quantitative losses depending only on the threshold value $\eps$ and the alphabet size $q$.
Our proof of this inequality is actually a simple generalization of an argument of \cite{BRS11}, recently reinterpreted in a closely-related setting by \cite{BHSV25}.

\subsection{Results}

\paragraph{Quadratic Optimization over the Hypercube.}
Our first result is an algorithm for quadratic optimization over the hypercube when the coefficient matrix $A$ has few negative eigenvalues. 

\begin{restatable}[Quadratic optimization]{theorem}{quadraticoptimization}
    \label{th2:quadratic-optimization}
	Let $A \in \R^{n \times n}$ be an arbitrary symmetric matrix, scaled to have $\opnorm{A} \le 1$, and let $\OPT = \max_{x \in \{\pm 1\}^n} x^\top A x$.
    Set $k = \rank_{\ge \eps}(A)$.
    Then, there is an algorithm running in time $\poly(n) \cdot \left( \frac{1}{\eps} \right)^{O(k)}$ that outputs a (random) vector $\wh{x} \in \{\pm 1\}^n$ such that
    \[ \E \wh{x}^\top A \wh{x} \ge \OPT - O(\eps n) \mper \]
\end{restatable}
\begin{remark}
    While we state~\cref{th2:quadratic-optimization} for the case of optimization over $\{\pm 1\}$, our main insights can be generalized to optimizing $x^\top A x$ over convex domains with efficient separation oracles when $A$ has few negative eigenvalues. 
\end{remark}

Note that if $A$ has no positive eigenvalues and has a $0$ diagonal, then $\max_{x \in \{ \pm 1\}^n} x^\top A x$ can be performed via maximization of a concave function over a convex set, since it suffices to optimize over $x$ in the solid cube $[-1,1]^n$ and round the result greedily.
Our result can be viewed as showing that even if the optimization problem $\max_{x \in [-1,1]^n} x^\top A x$ is not convex, but the non-convexity is restricted to a small-dimensional subspace, then we can still efficiently find an $\wh{x}$ which is close to optimal.
Intuitively, this is because one can perform limited brute-force search over the span of the positive eigenvectors of $A$, and for each point in the span solve an associated \emph{convex} optimization problem.

\paragraph{MAX-CUT on Low Threshold Rank Graphs.}

One special case of the quadratic optimization problem is when $A$ is the negation of the normalized adjacency matrix of a regular graph.
In this case, we have the MAX-CUT problem. In fact, the ideas in \Cref{th2:quadratic-optimization} also extend to irregular graphs, and the algorithm can be sped up using near-linear-time SDP solvers \cite{fastsdp2csps10}, to obtain the following.

\begin{restatable}[$(1+O(\eps))$ Approximation for MAX-CUT in $\poly(1/\eps)$ Time]{theorem}{maxcut}
    \label{th3:maxcut}
    Let $G$ be a graph on vertex set $[n]$ with $m$ edges, with adjacency matrix $A$ and diagonal degree matrix $D$.
    Let $\OPT$ be the size of the largest cut in $G$.
    Set $k = \rank_{\le -\eps}(D^{-1/2} A D^{-1/2})$.
    In time $ \left( \frac{1}{\eps} \right)^{O(k)} \cdot \wt{O} \left( n + m \right)$, one can output a (random) cut attaining value $\left( 1 - O(\eps) \right) \cdot \OPT$ in expectation.
\end{restatable}

We now compare \cref{th3:maxcut} to prior work.
The problem of MAX-CUT on low threshold rank graphs has been studied in several prior works, with various runtimes and notions of threshold rank. These are summarized in the table below. Note that all parameters $k$ are defined with respect to the normalized adjacency matrix.
The algorithms discussed in the table all produce $(1+O(\eps))$ approximations to the maximum cut.

\begin{table}[ht]
\begin{center}
\caption{\label{tab:threshold-rank}MAX-CUT on Low Threshold Rank Graphs}
\begin{tabular}{| c | c | c | } 
\hline
Paper & Runtime & Parameter $k$ \\
\hline
\cite{BRS11} & $2^{O(k/\eps^4)} \cdot n^{O(1)}$ & \# of eigs in $[\eps^2,1]$ \\
\cite{GS11} & $n^{O(k/\eps^2)}$ &  \# of eigs $\le -\eps$ \\
\cite{GS12} & $2^{k/\eps^3}\cdot n^{O(1/\eps)}$ & \# of eigs $\leq -\eps^2$\\
\cite{OGT13} + \cite{10.1145/3406325.3451126}\tablefootnote{\cite{OGT13} provides the running time $2^{\wt{O}(k^{1.5}/\eps^3)} + \poly(n)$ via a regularity decomposition of the graph. Using matrix multiplicative weights to find the regularity decomposition via a fast SDP solver yields the running time we quote. Matrix multiplicative weights-based SDP solvers appear first in \cite{arora2016combinatorial}; the variant for regularity decompositions appears in \cite{uniquedecodingtashma}.} & $2^{\wt{O}(k^{1.5}/\eps^3)} + \wt{O}(n + m)$ & sum of squares of eigs not in the range $[-\eps, \eps]$  \\
\hline
This work & $\left( \frac{1}{\eps} \right)^{O(k)} \cdot \wt{O} \left( n + m \right) $ & \# of eigs $\le -\eps$\\
\hline
\end{tabular}
\end{center}
\end{table}

Similar to \cite{GS11,GS12}, our runtime is parameterized by the number of eigenvalues of the (normalized) adjacency matrix which are smaller than $-\poly(\eps)$.%\footnote{Note that we have that for any adjacency matrix $A$ the number of eigenvalues less than $-\eps^2$ is at least the number of eigenvalues less than $-\eps$, so the parameter $k$ in~\cite{GS12} is larger than the parameter $k$ in both~\cite{GS11} and our work.}
However, compared to \cite{GS11,GS12}, our running time has an exponentially better dependence on $1/\eps$.
Regularity lemma-based algorithms for MAX-CUT on low threshold rank graphs such as \cite{OGT13} are instead parameterized with $k$ being as the sum of the squares of eigenvalues not in the range $[-\eps, \eps]$. Since all the eigenvalues of the normalized adjacency matrix lie in the range $[-1,1]$, this could be much smaller than the number of such eigenvalues.
However, the sum of the squares of the eigenvalues not in the range $[-\eps, \eps]$ is at least $O(\eps^{2}) \cdot \left(\text{\# of eigenvalues not in the range }[-\eps, \eps]\right)$; we also require only control on the negative eigenvalues of $A$ rather than positive and negative eigenvalues.
As a result, our exponential dependence on $k$ and $\eps$ also improve on regularity lemma-based approaches~\cite{OGT13, 10.1145/3406325.3451126}, although we incur an multiplicative dependence in runtime on the $\wt{O} \left( n + m \right)$, $\left( \frac{1}{\eps} \right)^{O(k)}$ terms, as opposed to an additive dependence. 

\begin{comment}
\begin{restatable}[Quadratic optimization over general convex domains]{theorem}{generaldomains}
    Let $\Omega \subseteq R \bbB^{n} \subseteq \R^n$ be a convex domain that admits an efficient (polynomial-time) separation oracle.
    Let $M \in \R^{n \times n}$ be a symmetric matrix such that $\opnorm{M} \le 1$.
    Set $\OPT = \sup_{x \in \Omega} x^\top M x$.
    Then, denoting $k = \rank_{\ge \eps}(M)$, there exists an algorithm running in time $\poly(n) \cdot \left( \frac{1}{\eps} \right)^{k}$ that finds a $\wh{x} \in \R^n$ such that
    \begin{itemize}
        \item $\inf_{x \in \Omega} \|\wh{x} - x\|_2 \le ?$, and
        \item $\wh{x}^\top M \wh{x} \ge \OPT - O(\eps R^2)$.
    \end{itemize}
\end{restatable}
\end{comment}

\paragraph{General 2CSPs.} In addition to fast algorithms for MAX-CUT, we give the first near-linear time algorithm for general CSPs on low threshold rank graphs.

\begin{restatable}{theorem}{arbitrarycsps}
    \label{th1:arbitrary-csps}
    Let $G$ be a graph on $[n]$ with $m$ edges, with adjacency matrix $A$ and degree diagonal $D$.
    For each edge $ij$ in $G$, let $\varphi_{ij} : [q]^2 \to \{0,1\}$ be a predicate on alphabet size $q$ that is not identically $0$.
    Also let $\OPT = \max_{x \in [q]^n} \Phi(x) \defeq \max_{x \in [q]^n} \sum_{ij \in E(G)} \varphi_{ij}(x_i,x_j)$.
    Set $k = \rank_{\ge \eps^2} (D^{-1/2} A D^{-1/2})$.
    In time $\left( \frac{1}{\eps} \right)^{ O \left( k / \eps^4 \right) } \cdot \wt{O} \left( n + m\right) \cdot \poly(q)$, one can find a random assignment $\wh{x}$ with $\E \Phi(\wh{x}) \ge \OPT - O(\eps q m)$.
\end{restatable}

%To obtain Theorem~\ref{thm:arbitrary-csps}, we use a proof technique of \cite{BRS11,BHSV25} to prove an inequality relating the threshold rank of $D^{-1/2} A D^{-1/2}$ to the threshold rank of the \emph{label-extended graph} of $\Phi$.

The main prior work that gave algorithms for general CSPs under threshold rank conditions is~\cite{BRS11}. Our work matches the exponential dependence on $k, \eps$ (up to factors of $\log(1/\eps)$) of~\cite{BRS11} while improving the polynomial dependence on $n$ and $m$ to near-linear.
The main barrier to achieving a near-linear running time via prior algorithmic techniques was that the algorithmic approaches that lead to $\wt{O}(n+m)$ time (e.g., regularity lemmas \cite{OGT13,10.1145/3406325.3451126}) also lead to exponential dependence on the threshold rank of the label-extended graph (see~\cref{sec:prelim}).
Our main technical contribution to overcome this barrier is to bound the threshold rank of the label-extended graph in terms of the threshold rank of $G$; combining this inequality with fast SDP solvers using matrix multiplicative weights leads to the algorithm of~\cref{th1:arbitrary-csps}.

% While the key ingredients of near-linear time regularity decompositions were available in prior work, no work before ours obtains a nearly linear-time algorithm for general 2CSPs under threshold rank conditions.
\paragraph{Which Graphs Have Small Threshold Rank?}
As noted in prior work, many interesting families of graphs have low threshold rank.
Random graphs with degree $d$ have threshold rank $1$ for $\eps > c/\sqrt{d}$; more generally, one-sided spectral expanders have positive threshold rank $1$, and two-sided spectral expanders have negative threshold rank equal to zero.

What about graphs beyond expanders?
Focusing on the positive threshold rank case, graphs in which all small sets are highly expanding have bounded threshold rank \cite{steurer2010subexponential}.
And higher-order Cheeger inequalities like \cite{10.1145/2665063} say that if a graph has no sparse $k$-way cut, then the $k$-th largest eigenvalue of the normalized adjacency matrix is bounded from above, which yields a bound on the threshold rank.
Thus, our algorithm for 2CSPs is efficient on graphs satisfying this combinatorial property.

Which graphs have small negative threshold rank---the setting of our MAX-CUT algorithm---is less explored.
Intuitively, a graph $G$ with small $\rank_{\leq -\eps}(A)$ should have ``few'' cuts which cut significantly more than half of the edge set.
But to our knowledge, \emph{dual} Cheeger inequalities like \cite{doi:10.1137/090773714} and higher-order dual Cheeger inequalities like \cite{LIU2015306} only describe combinatorial properties of $G$ which imply bounds on $\rank_{\leq -(1-\eps)}(A)$, the number of eigenvalues close to $-1$, which could be much smaller than the number of eigenvalues less than $-\eps$.
Thus we pose the following problem: \emph{provide a combinatorial characterization of graphs $G$ with few eigenvalues less than $-\eps$.}

Returning to general 2CSPs, if we apply a standard reduction \cite{10.1109/CCC.2010.20} from a general 2CSP to a quadratic optimization problem, we get running times in terms of the threshold rank of the so-called \emph{label-extended} graph of the CSP.
One of our key technical insights is in relating the threshold rank of this graph to the threshold rank of the constraint graph $A$.
We discuss this in more detail in~\cref{sec:label-extended-threshold-rank}. A separate argument shows that for dense graphs $G$, the threshold rank of the label-extended graph is at most $O(q^2/\eps^2)$, meaning that our algorithm also recovers the best known running time for a PTAS for dense 2CSPs on constant-sized alphabets~\cite{yaroslavtsev2014going,yoshida2014approximation}.

\subsection{Techniques}

\paragraph{Algorithm for Boolean Quadratic Optimization and MAX-CUT.}
Our algorithm is a novel combination of prior SDP-based approaches and subspace enumeration techniques. We will sketch the proof of~\cref{th2:quadratic-optimization} in the case where $A$ has $\operatorname{diag}(A) = \mathbf{0}$.\footnote{This assumption is without loss of generality, since the diagonal contributes a constant amount to the objective on any $x \in \{\pm 1\}^n$.}
Recall that the canonical SDP relaxation of the problem $\max_{x \in \{\pm 1\}^n} x^\top A x$ is as follows:
\[
\begin{aligned}
    \max_{X \in \R^{n \times n}, X \succeq 0} \quad & \langle A, X \rangle \\
    \text{s.t.} \quad & \operatorname{diag}(X) = \mathbf{1}.
\end{aligned}
\]
Equivalently, for reasons we will see in a moment, we could optimize over $X \in \R^{n+1 \times n+1}$:
\[
\begin{aligned}
    \max_{X \in \R^{n+1 \times n+1}, X \succeq 0} \quad & \langle A, X_{[1,n+1],[1,n+1]} \rangle \\
    \text{s.t.} \quad & \operatorname{diag}(X) = \mathbf{1} \, ,
\end{aligned}
\]
where $X_{[1,n+1],[1,n+1]}$ is the bottom right $n\times n$ submatrix of $X$.

A solution to the second relaxation can be interpreted as an linear operator $\pE \, : \, \R[x_1,\ldots,x_n]_{\leq 2} \rightarrow \R$, which maps polynomials in $x_1,\ldots,x_n$ of degree at most $2$ to scalars; we call $\pE$ a \emph{pseudoexpectation}: we can take $\pE x_i = X_{0,i}$ and $\pE x_i x_j = X_{i,j}$ and extend linearly.
In this language, the SDP above becomes $\max_{\pE} \pE \iprod{x, Ax}$.
While we will not use the full machinery associated with this point of view and the corresponding Sum-of-Squares semidefinite hierarchy, we still use the pseudoexpectation language for ease of exposition.
Note that the PSD constraint ensures that the pseudoexpectation of any square polynomial is non-negative, and notably that $\pE xx^\top - (\pE x)(\pE x)^\top$ is a PSD matrix.

As a first attempt at an algorithm for quadratic optimization, consider the following naive rounding scheme applied to the optimal $\pE$ -- 
round each $x_i$ independently by flipping a $\pm 1$-valued coin with expectation $\pE x_i$.
Of course, by symmetry there is an optimal solution $\pE$ for which $\pE x_i = 0$ for all $i$, meaning that this rounding scheme will not give any nontrivial guarantees unless we modify the SDP.
But, to see how to modify the SDP, it will be helpful to further analyze this rounding scheme.

Note that the solution produced has (in expectation) objective value $(\pE x)^\top A (\pE x)$ and thus we can write difference between the objective value of the SDP and the rounded objective value as:
\[ \langle A, \pE xx^\top \rangle - \langle A, (\pE x) (\pE x)^\top \rangle = \langle A, \wt{\Cov} \rangle\,,\]
where $\wt{\Cov} = \pE xx^\top - (\pE x)(\pE x)^\top$.
We call this quantity the \emph{rounding error}.
Our goal is to bound it from above.

We can decompose the rounding error in two parts: that incurred on the subspace spanned by the eigenvalues of $A$ that are less than $\eps$ (which we refer to as $A_{< \eps}$), and that incurred on the part of $A$ with eigenvalues at least $\eps$ (which we refer to as $A_{\ge\eps}$): 
\[
    \langle A, \wt{\Cov} \rangle = \iprod{A_{< \eps}, \wt{\Cov}} + \iprod{A_{\ge \eps}, \wt{\Cov}} \,.
\]

First off, the $\langle A_{<\eps}, \wt{\Cov} \rangle$ part of the rounding error cannot be too large for any SDP solution $\pE$.
Since $\wt{\Cov}$ is positive semidefinite, we can bound $\iprod{A_{< \eps}, \wt{\Cov}} \leq \eps \cdot \Tr(\wt{\Cov}) \leq \eps n$.
Thus, for the rounding to be successful -- incurring error at most $O(\eps n)$ -- we only require that the SDP solution $\pE$ satisfies $\langle A_{\ge \eps}, \wt{\Cov} \rangle \leq \eps n$. 
However, in general, the SDP solution associated with this optimization problem might not have this property.

This brings us to a modified SDP relaxation of $\max_{x \in \{ \pm 1\}^n} x^\top A x$.
Suppose that we knew a vector $v$ such that $\|v - \Pi x^*\|^2 \leq \eps n$, where $\Pi$ is the projector to image of $A_{\geq \eps}$ and $x^*$ is an optimal solution to $\max_{x \in \{ \pm 1 \}^n} x^\top A x$.
This is not too much to assume, since by enumerating over at most $(1/\eps)^{O(\rank_{\geq \eps}(A))}$ vectors $v$ in the subspace $\text{Im}(A_{\geq \eps})$, we can guess such a $v$.
Then the following SDP would still be a relaxation:
\begin{equation}
    \label{eq:techniques-1}
    \max \pE \iprod{x, Ax} \text{ s.t. } \pE \|\Pi x - v\|^2 \leq \eps n \, .
\end{equation}

Crucially, if $\pE$ is any pseudoexpectation for which there exists $v$ such that $\pE \| \Pi x - v \|^2 \leq \eps n$, then
\[
    \iprod{A_{\geq \eps}, \wt{\Cov}} = \iprod{A_{\geq \eps}, \Pi \wt{\Cov} \Pi} \leq \|A_{\geq \eps}\| \cdot \Tr \Pi \wt{\Cov} \Pi \leq \pE \|\Pi (x - \pE x)\|^2 \, .
\]
Furthermore, using triangle inequality and PSD-ness of $\pE$, we have
\[
    \pE \| \Pi(x - \pE x)\|^2 \leq 2 (\pE \|\Pi x - v\|^2 + \|v - \pE \Pi x\|^2) \leq 4 \pE \|\Pi x - v\|^2 \, .
\]
Putting it together, we get that for any feasible $\pE$ for \eqref{eq:techniques-1}, the rounding error is at most $O(\eps n)$.

This brings us to our final algorithm for Boolean quadratic optimization.
Given $A$, enumerate the image of $A_{\geq \eps}$.
For each $v$ in the enumeration, solve \eqref{eq:techniques-1} and round by sampling coordinates independently.
Output the best solution found in this way.

Since (up to an additive shift and rescaling) MAX-CUT on regular graphs is just $\max x^\top (-A) x$, where $A$ is the normalized adjacency matrix, this gives our $\poly(1/\eps)$ time algorithm for MAX-CUT on regular graphs with few negative eigenvalues (\cref{th3:maxcut}). We note that this argument can be generalized to irregular graphs and we defer this to~\cref{sec:quadratic-opt}.

%Here, we observe that since $A_{\ge \eps}$ corresponds to a low-dimensional subspace (specifically one with dimension $k$), we can solve a system of not too many SDPs where each one satisfies $\langle A_{\ge \eps}, \wt{\Cov} \rangle \leq \eps n$ and at least one has a large objective value. Note that if $\Pi$ is the projection to the large eigenspaces of $A$ and $\Pi x$ has small pseudo-covariance, or is approximately fixed, then we have that $\langle A_{\ge \eps}, \wt{\Cov} \rangle$ must be small since $A$ has bounded operator norm. Thus, it suffices to know the ``correct'' value of $\Pi x^*$ for the optimal $x^*$, since if we were given this vector $v$ we could simply enforce the constraint $\pE \norm{\Pi x - v}_2^2 \leq \eps n$ which is now convex for a fixed $v$.

%However, there are not too many possible values of $v$, since it is enough to know $v$ up to $O(\sqrt{\eps n})$ accuracy! In fact, since $v$ lives in a subspace of dimension $k$ and has norm at most $\sqrt{n}$, we have that by standard net arguments that there is a set of $\left(\frac{1}{\eps}\right)^{O(k)}$ many vectors $v_j$ such that at least one $v_j$ is within $\sqrt{\eps n}$ distance of the correct $v$. Thus, we can enforce the above constraint for this $v_j$ without any loss in objective value, since the pseudo-distribution supported only on the optimal $x^*$ is still a feasible solution. 

\paragraph{Algorithm for General 2CSPs, and Threshold Ranks of Label-Extended Graphs}
Now we move on to our algorithm for general 2CSPs with alphabet size $[q]$.
It is by now standard to associate an $n$-variable 2CSP instance $\varphi$ with variables $x_1,\ldots,x_n$ and constraints $\varphi_{ij} \, : \, [q] \times [q] \rightarrow \{0,1\}$ a \emph{label-extended graph} $M$ on vertex set $[n] \times [q]$, where $(i,a)$ is adjacent to $(j,b)$ if $x_i = a, x_j = b$ satisfies $\phi_{ij}$ \cite{10.1109/CCC.2010.20}.
Using largely the same algorithm and argument as in the Boolean case, we can obtain a $(1+O(\eps))$ approximation algorithm in time $(1/\eps)^{O(\rank_{\geq \eps}(M))} \cdot (n)^{O(1)}$.

The prior work \cite{BRS11} gives a $(1+O(\eps))$ approximation algorithm for 2CSPs with running time depending only on the threshold rank of the base graph $G$; their algorithm runs in time $\exp(\tilde{O}(\rank_{\geq \eps^2}(G)/\eps^4) (n)^{O(1)}$.
We recover this running time (and later, unlike \cite{BRS11}, speed up to near-linear time) by proving a new inequality relating $\rank_{\geq \eps}(M)$ to $\rank_{\geq \eps^2}(G)$.
Namely, we show that $\rank_{\geq 2 q \eps}(M) \leq \rank_{\geq \eps^2}(G)/\eps^4$ (\cref{cor:label-extended-threshold-rank-bound}).
We actually observe that this bound follows from a small adaptation of an argument by \cite{BHSV25}, itself adapted from \cite{BRS11}.

\paragraph{From Polynomial to Near-Linear Time}
The main contributor to the running time of our algorithm is the time to solve the SDPs. 
To speed this up from $(mn)^{O(1)}$ time to $\tilde{O}(m+n)$, we rely on the matrix multiplicative weights framework as developed in \cite{fastsdp2csps10}, which gives a near-linear time SDP solver for the basic SDP relaxation of CSPs.
To adapt this approach to our setting, we just need to enforce the additional constraint $\pE \| \Pi x - v\|^2 \leq \eps n$ when solving the SDP.
This turns out not to be too difficult.
The key quantity governing the difficulty of enforcing a constraint in the multiplicative weights framework is its ``width''; in this case, the operator norm of the matrix which forms the constraint.
Enforcing this extra constraint turns out to be possible with only $\poly(q/\eps)$ width, meaning that the framework of \cite{fastsdp2csps10} can solve our SDPs in time $\poly(q/\eps) \cdot \tilde{O}(m+n)$.

\subsection{Related Work}

\paragraph{SDP-based algorithms.}
Semidefinite programming and SDP hierarchies have been used to give algorithms for CSPs under varying threshold rank conditions~\cite{BRS11,GS11, GS12}. The work of~\cite{BRS11} achieves similar approximation guarantees to that of our work, but incurs a polynomial as opposed to near-linear dependence in runtime on the size of the input. For restricted classes of CSPs (including MAX-CUT), the work of Guraswami--Sinop \cite{GS11, GS12} yields algorithms with similar approximation guarantees. Furthermore, in the case of 3-coloring, recent work by Hsieh utilizes SDP-based approaches to find proper 3-colorings on almost half the edges in 3-colorable graphs with small threshold rank~\cite{hsieh2025coloring3colorablegraphslow}.

\paragraph{Subspace enumeration algorithms.}
Subspace enumeration has been an influential algorithm design technique for CSPs, specifically Unique Games, and related problems such as small set expansion~\cite{10.1109/CCC.2010.20, ABS15}. These works give algorithms whose runtime depends on the threshold rank of the \emph{label-extended} graph. A key challenge in these approaches is bounding this quantity, since the label-extended graph having bounded threshold rank is a stronger condition than the constraint graph having bounded threshold rank. We discuss threshold rank bounds further in the paragraph below.

\paragraph{Regularity lemmas.}
Regularity lemmas have been a key tool for giving fast algorithms for constraint satisfaction and related problems. Their use dates back to early approximation schemes for dense CSPs in the 90s and early 2000s~\cite{fk1999, 10.1145/509907.509945}, and more recent work extended these ideas by giving a regularity lemma for graphs with low threshold rank~\cite{OGT13}. Extensions of this work to higher-arity CSPs have also been influential in the design of fast decoders for error correcting codes with high rate~\cite{10.1145/3406325.3451126}.

\paragraph{Dense CSPs.} Many prior works have studied approximation schemes for CSPs on graphs which are either dense or pseudo-dense~\cite{fk1999, 10.1145/509907.509945, doi:10.1137/070709529,doi:10.1137/080730160}. Note that all dense and psuedo-dense graphs have low threshold rank, and thus our results and prior results on approximation schemes for graphs with low threshold rank apply to a larger class of instances. 

\paragraph{Dual higher-order Cheeger inequalities.}

Higher-order dual Cheeger inequalities \cite{doi:10.1137/090773714, LIU2015306} give a relationship between the magnitude of $\lambda_{n-k}$ and combinatorial properties of the graph. In particular, they can be used to show that graphs with this combinatorial property---that all vertex-disjoint non-empty subsets $V_1, \ldots, V_k$ must have at least one $V_i$ which is far from being a bipartite connected component of the original graph---have eigenvalues other than $\lambda_{n-k}, \ldots, \lambda_n$ bounded away from $-1$. However, this does not yield a combinatorial class of graphs which our algorithm runs fast on since higher-order dual Cheeger inequalities are too quantitatively loose to give bounds on $\rank_{\le -\eps}$.\footnote{This is due to the fact that the inequality becomes trivial when considering eigenvalues which are negative but very close to $0$.}

\paragraph{Threshold rank bounds.} As discussed above, a line of prior work yielded algorithms for Unique Games which runs in time parameterized by the threshold rank of the label-extended graph~\cite{10.1109/CCC.2010.20, ABS15}. Such work also studied the relationship between the threshold rank of the constraint graph and the threshold rank of the label-extended graph. However, existing prior work all lose polynomial factors in $n$~\cite{ABS15}. 

%% file: content/prelims.tex
%!TEX root = ./main.tex

\section{Preliminaries}
\label{sec:prelim}

\paragraph{Threshold Rank and Label Extension.} We first formally define the threshold rank of a matrix as well as the label extended graph of a CSP.

\begin{definition}[Threshold Rank]
    Let $A \in \R^{n \times n}$ be a matrix with real eigenvalues, and let $0 < \eps < 1$. The (one-sided) $\eps$-threshold rank of $A$, denoted $\rank_{\geq \eps}(A)$ or $\rank_{\le -\eps}$, is the number of eigenvalues of $A$ which are at least $\eps$ or less than $-\eps$ respectively. 
    % Similarly, the two-sided $\eps$-threshold rank of $A$, denoted $\rank^\circ_{\ge \eps}(A)$,\amit{just some placeholder notation, ff to change} is the number of eigenvalues of $A$ that have magnitude at least $\eps$ (but are possibly negative).
\end{definition}

\begin{definition}[Label Extended Graph]
    Let $G$ be a graph on $[n]$ with adjacency matrix $A$. For each edge $ij$ in $G$ let $\varphi_{ij} : [q]^2 \to \{0,1\}$ be a predicate on alphabet size $q$. The label extended graph $G'$ is a graph with vertices $\{v_{i, \alpha}\}_{i \in [n], \alpha \in [q]}$ such that there is an edge between $(i,\alpha)$ and $(j, \beta)$ if and only if $ij \in G$ and $\phi_{ij}(\alpha,\beta) = 1$.
\end{definition}

\paragraph{Semidefinite programming and pseudodistributions.}

Although we do not use the full power of the sum-of-squares programming hierarchy, we will still use the pseudoexpectation notation for ease of exposition.

\begin{definition}[Pseudodistribution]
    A \emph{pseudodistribution} $D$ of degree $t$ is a function from $\mathbb{R}^n$ to $\mathbb{R}$ with finite support such that $\sum_{x \in \operatorname{supp}(D)} D(x) = 1$ and $\sum_{x \in \operatorname{supp}(D)} D(x) p(x)^2 \geq 0$ for all polynomials $p(x)$ with $\operatorname{deg}(p(x)^2) \leq t$.
\end{definition}

\begin{definition}[Pseudoexpectation]
    Given a pseudodistribution $D$ of degree $t$, the associated \emph{pseudoexpectation} $\pE_{D}$ is the linear operator from the space of functions to $\R$ that is defined by $f$ mapping to $\pE_{D} f(x) \defeq \sum_{x \in \operatorname{supp}(D)} D(x) f(x)$.
\end{definition}

The non-negativity of the pseudoexpectation on squared polynomials implies the following easy fact.

\begin{fact}
    Let $\pE$ be a degree $t$ pseudoexpectation and let its pseudocovariance matrix be $\wt{\Cov} = \pE xx^\top - \left(\pE x \right)\left(\pE x\right)^\top$. Then $\wt{\Cov} \succeq 0$.
\end{fact}

\begin{definition}[Constrained pseudodistributions]
\label{def:constrained-pseudo-distributions}
Let $\calA = \Set{ p_1\geq 0 , p_2\geq0 , \dots, p_r\geq 0}$ be a system of $r$ polynomial inequality constraints of degree at most $d$ in $m$ variables.
Let $\mu$ be a degree-$\ell$ pseudodistribution over $\mathbb{R}^m$.
We say that $\mu$ \emph{satisfies} $\calA$ at degree $\ell \ge1$ if for every subset $\calS \subset [r]$ and every sum-of-squares polynomial $q$ such that $\deg{q} + \sum_{i \in \calS } \max\Paren{ \deg{p_i}, d} \leq \ell$, $\pE_{\mu}{ q \prod_{i \in \calS} p_i } \geq 0$.
Further, we say that $\mu$ \emph{approximately satisfies} $\calA$ if the above inequalities are satisfied up to additive error $\pE_{\mu}{ q \prod_{i \in \calS} p_i } \geq -2^{-n^{\ell} } \norm{q} \prod_{i \in \calS} \norm{p_i}$, where $\norm{\cdot}$ denotes the Euclidean norm of the coefficients of the polynomial, represented in the monomial basis.  
\end{definition}

Note that when we have a pseudodistribution of degree $2$ then there is an easy weak separation oracle for the convex set of moment tensors of constrained pseudodistributions via checking that: (1) the moment matrix is PSD, and (2) $\pE p_i \geq 0$ for all $p_i \in \calA$\footnote{There is also a weak separation oracle for pseudoexpectation of higher degree, see~\cite{shor1987approach, parrilo2000structured}.}. Furthermore, if $ \calA$ is \emph{explicitly bounded}, in that it contains a constraint of the form $\{ \|x\|^2 \leq 1\}$, the above observation in combination with \cite{grotschel1981ellipsoid} yields the following fact.

\begin{theorem}[Efficient optimization over pseudodistributions of degree $2$]
    \label{fact:eff-pseudo-distribution-deg-2}
    There exists an $(m+r)^{O(1)} $-time algorithm that, given any explicitly bounded and satisfiable system $ \calA$ of $r$ polynomial constraints in $m$ variables, outputs a degree-$2$ pseudodistribution that satisfies $ \calA$ approximately, in the sense of~\cref{def:constrained-pseudo-distributions}.\footnote{
    Here, we assume that the bit complexity of the constraints in $ \calA$ is $(m+t)^{O(1)}$.
}
\end{theorem}

% \prashanti{change to be boolean vectors}
% \begin{lemma}
%     \label{lem:net-bound-size}
%     Let $D \in \R^{n \times n}$ be an arbitrary positive semidefinite matrix, and $\Pi \in \R^{n \times n}$ an arbitrary rank-$k$ projection matrix. Then, there exists a set $S \subseteq \R^k$ such that for any $u \in \{\pm 1\}^n$, there exists $v \in S$ such that $\| \Pi D^{1/2} u-v \|_2^2 \le \eps \Tr D$, and $|S| \le \left(\frac{1}{\eps}\right)^{O(k)}$.
% \end{lemma}
% \begin{proof}
%     Note that for any $u \in \{\pm 1\}^n$, $\|\Pi D^{1/2} u\|_2^2 \le \Tr D$, that is, $\Pi D^{1/2} \{\pm 1\}^n \subseteq \sqrt{\Tr D} \cdot \bbB^{k}$. We may now simply take $S$ to be a net of $\sqrt{\Tr D} \cdot \bbB^{k}$ of fineness $\sqrt{\eps \Tr D}$, that is of size $O\left(\frac{1}{\eps}\right)^{k/2}$, as claimed. \amit{maybe add a citation to the (easy) net size bound?}
% \end{proof}

%% file: content/large-alphabet.tex
%!TEX root = ./main.tex

\section{Algorithm and Analysis}
\label{sec:quadratic-opt}

\subsection{Improving the runtime dependence on $\eps$}

In this section, we shall describe our primary algorithm, with a running time depending on the threshold rank of the label-extended graph.

\begin{theorem}
	\label{th:main-body}
	Let $M \in \R^{nq \times nq}$ be a symmetric matrix such that $\opnorm{M} \le 1$, and the $n$ diagonal $q \times q$ blocks of $M$ are all identically $0$. Also, let $D \in \R^{n\times n}$ be an arbitrary diagonal matrix with non-negative entries and $E = D \otimes \Id_q$. Let $k = \rank_{\ge \eps}(M)$. Consider the optimization problem
	\[ \OPT \defeq \max_{y \in \calC_{q}^n} \Phi(y) = \max_{y \in \calC_{q}^n} y^\top E^{1/2} M E^{1/2} y \mcom \]
	where $\calC_{q}^n$ is the subset of $y \in \{0,1\}^{nq}$ such that each of the $n$ blocks of $y$ (each of size $q$) has exactly one $1$.
	There exists an algorithm running in time $\poly(n) \cdot \left( \frac{1}{\eps} \right)^{O(k)}$ that returns a random $y \in \calC_{q}^{n}$ with
	\[ \E \Phi(y) \ge \OPT - O\left(\eps \Tr D \right) \mper \]
\end{theorem}

Given this, our result on quadratic optimization over the hypercube follows near-immediately. We restate it for convenience.

\quadraticoptimization*

\begin{proof}
	We may assume without loss of generality that the diagonal entries of $A$ are all $0$, since they contribute a constant when optimizing over the hypercube. The result then immediately follows from \Cref{th:main-body} on setting $q = 2$, $D = \Id_{n}$, and $M = A \otimes \begin{pmatrix} 1 & -1 \\ -1 & 1 \end{pmatrix}$.
\end{proof}

The algorithm for \Cref{th:main-body} is described in \Cref{alg:fast-csp}.

\begin{algorithm}[ht]
	\DontPrintSemicolon
	\KwIn{$M \in \R^{nq \times nq}$ with $\rank_{>\eps}(M) = k$ and $\opnorm{M} \le 1$, and PSD diagonal $D \in \R^{n \times n}$}
	% \KwOut{???}
	$\Pi \gets $ projection onto the top $k$-dimensional eigenspace of $M$\;
	$\calS \gets \sqrt{\eps \cdot \Tr D}$-net of $\sqrt{\Tr D} \cdot \bbB^{k}$.\;
    $E \gets D \otimes \Id_q$.\;
	\For{$v \in \calS$} {
		Find a degree-$2$ pseudoexpectation $\pE$ over $[q]^n$ optimizing $\pE y^\top E^{1/2} M E^{1/2} y$, subject to the constraints $y_i(1-y_i) = 0$, $\sum_{j \in [q]} y_{qi + j} = 1$, and $\pE \| \Pi E^{1/2} y - v \|_2^2 \le \eps \Tr D$.\;
        Round $\pE$ to a (random) $y_v \in [q]^n$ by independently sampling each index $(y_v)_i$ according to the marginals prescribed by $\pE$.\;
    	% Round $\pE$ to a (random) $x_v \in [q]^n$ by independently sampling each coordinate according to its marginal.\; 
	}
	\Return{the $y_v$ that maximizes $\Phi(y_v)$}\;
	\caption{Solving CSPs faster on bounded threshold rank graphs}
	\label{alg:fast-csp}
\end{algorithm}

% It solves multiple SDPs, one for each $v$. Compare this to the usual global correlation rounding paradigm, which finds a single global (high-degree) pseudodistribution, then performs some post-processing to obtain another pseudodistribution, which is then rounded to a product distribution.

For the remainder of this section, suppose one has a pseudoexpectation $\pE$ over boolean variables $(y_{i,\alpha})_{i \in [n], \alpha \in [q]}$, satisfying the constraint that exactly one of the $(y_{i,\alpha})_{\alpha \in [q]}$ is $1$ for every $i$ (this enforces the requirement that $y \in \calC_{q}^n$).
% Note that for a $y$ of this form with corresponding $x \in [q]^n$, we have $\Phi(x) = y^\top A y$. We abuse notation to denote $\Phi(y) = y^\top A y$. 

Finally, consider the product distribution $\mu^{\otimes}$ over $[q]^n$ such that $\mu\left( \alpha_1 , \ldots , \alpha_n \right) = \prod_{j \in [n]} \pE y_{j,\alpha_j}$. We also have the associated distribution $\nu$ over $\calC_q^n$ induced by $\mu^{\otimes}$: in other words, the different blocks are independent, and within each block, the probabilities of the different symbols are given by the marginals under the pseudodistribution.

\begin{lemma}
	\label{lem:error-bound}
	Let $\Pi$ be the projector onto the eigenspace of $M$ with eigenvalues larger than $\eps$, and $v \in \R^{nq}$ an arbitrary vector. Then,
	\[ \pE \Phi(y) - \E_{\by \sim \nu} \Phi(\by) \le \eps \Tr D + 4 \pE \left\| \Pi E^{1/2} (y - v) \right\|_2^2 \mper \]
\end{lemma}
\begin{proof}
	Because the diagonal blocks of $M$ are $0$, we are interested in
	\[ \pE y^\top E^{1/2} M E^{1/2} y - \left( \pE y \right)^\top E^{1/2} M E^{1/2} \left( \pE y \right) = \left\langle E^{1/2} M E^{1/2} , \wt{\Cov} \right\rangle \mcom \]
	where $\wt{\Cov} = \pE yy^\top - (\pE y)(\pE y)^\top$ is the ``pseudocovariance'' of $\pE$. We now decompose $A$ into its different eigenspaces and bound each separately. We have
	\begin{align*}
		\langle E^{1/2} M E^{1/2} , \wt{\Cov}\rangle &= \langle M_{<\eps} , E^{1/2} \wt{\Cov} E^{1/2} \rangle + \langle M_{\ge\eps} , E^{1/2} \wt{\Cov} E^{1/2}\rangle \\
			&\le \eps \Tr \left(E^{1/2} \wt{\Cov} E^{1/2}\right) + \langle M_{\ge \eps} , E^{1/2} \wt{\Cov} E^{1/2}\rangle \\
			&\le \eps \Tr \left(E^{1/2} \wt{\Cov} E^{1/2}\right) + \Tr\left( \Pi E^{1/2} \wt{\Cov} E^{1/2} \Pi \right).
	\end{align*}
    Note that the first term is
    \[ \eps \left\langle E , \wt{\Cov} \right\rangle = \eps \sum_{\substack{i \in [n] \\ \alpha \in [q]}} D_{ii} \pE y_{(i,\alpha)}^2 \leq \eps \sum_{i \in [n]} D_{ii} \sum_{\alpha \in q} \pE y_{(i,\alpha)} = \eps \Tr(D) \mcom \]
    utilizing that $\pE y_{qi+j}^2 \leq \pE y_{qi+j}$ and $\sum_{j \in [q]} \pE y_{qi+j} = 1$.
	Let $m = \pE y$. The second trace may be bounded as
	\begin{align*}
		\Tr\left( \Pi E^{1/2} \wt{\Cov} E^{1/2} \Pi \right) &= \pE \left\| \Pi E^{1/2} y - \Pi E^{1/2} m \right\|^2 \\
			&\le 2 \cdot \left( \pE \left\| \Pi E^{1/2} (y - v) \right\|_2^2 + \|\Pi E^{1/2} (m - v)\|^2 \right) \\
			&\le 4 \cdot \pE \left\| \Pi E^{1/2} (y - v) \right\|_2^2 \mcom
	\end{align*}
	where the last inequality is Jensen's, and the first is the almost-triangle inequality.
\end{proof}

We may now prove the result in this section with suboptimal running time.

\begin{proof}[{Proof of \Cref{th:main-body}}]
	We first observe that there exists some $v \in \calS$ such that the associated optimizing pseudoexpectation has objective value at least $\OPT$: indeed, consider an optimizer $y \in \calC_q^n$, and observe that $\|\Pi E^{1/2} y\|_2^2 \le \|D^{1/2} \mathbf{1}_n\|_2^2 \leq \Tr D$ since $y$ has exactly one non-zero entry in each $q$-length block, so there indeed exists a $v \in \calS$ such that $\| \Pi E^{1/2} y - v \|_2^2 \le \eps \cdot \Tr D$. Then, the Dirac distribution supported only on $y$ has objective value $\OPT$ and satisfies our constraints for this $v$.
    
	\Cref{lem:error-bound} implies that \Cref{alg:fast-csp} returns a point $\bx$ such that $\E \Phi(\bx) \ge \OPT - O(\eps \Tr D)$, as claimed. For the running time guarantee, we may use a standard SDP solver that runs in $\poly(n)$ time to find the pseudoexpectation of interest for each $v$. Thus, the running time is just $|\calS| \cdot \poly(n)$. The bound on the size of $\calS$, and thus the running time, is folklore.
\end{proof}

\subsection{Improving the runtime's dependence on $n$ for 2CSPs}
\label{subsec:fast-sdp}
In this section, we describe how to improve the $\poly(n,q) \cdot \left( \frac{1}{\eps} \right)^{\widetilde{O}(k)}$ runtime in the previous section to a near-linear time algorithm when $M$ corresponds to the normalized adjacency matrix of a label-extended graph. This yields a near-linear time algorithm for arbitrary 2CSPs when the threshold rank of the label-extended graph is small\footnote{In~\cref{sec:label-extended-threshold-rank} we will give runtime guarantees in terms of the threshold rank of the constraint graph by comparing these two quantities.}. Furthermore, with a small modification, we will show how to extend this to give a near-linear time algorithm for MAX-CUT on graphs where the normalized negative adjacency matrix has small threshold rank. 

\begin{theorem}[Fast Algorithm for 2CSPs with Small Label-Extended Threshold Rank]
\label{thm:near-linear-2csp-label-extended}
    Let $G$ be a graph on $[n]$ with $m$ edges, with adjacency matrix $A$ and degree diagonal $D$.
    For each edge $ij$ in $G$, let $\varphi_{ij} : [q]^2 \to \{0,1\}$ be a predicate on alphabet size $q$ that is not identically $0$ and let $M$ be the label-extended graph associated with the constraint graph $A$ and constraints $\varphi_{ij}$. Let $E = D \otimes \Id_q$.
    Also let $\OPT = \max_{x \in [q]^n} \Phi(x) \defeq \max_{x \in [q]^n} \sum_{ij \in E(G)} \varphi_{ij}(x_i,x_j)$.
    Set $k = \rank_{\ge \eps} (E^{-1/2} M E^{-1/2})$.
    Then,~\cref{alg:fast-csp} can be implemented in time $\wt{O} \left( n + m \right) \cdot \left(\frac 1 \eps\right)^{O(k)} \cdot \poly(q)$ and finds an assignment $\wh{x}$ with $\Phi(\wh{x}) \ge \OPT - O(\eps q^2 m)$.
\end{theorem}

\begin{comment}
    Note that max-cut corresponds to quadratic optimization over the hypercube with respect to the negative adjacency matrix $-A$---indeed, for any $x \in \{\pm 1\}^n$, the size of the corresponding cut is $\frac{1}{2} \cdot x^\top \left( D - A \right) x = \frac{m}{2} - \frac{1}{2} x^\top A x$. The desideratum immediately follows on instantiating \Cref{th:main-body} with $q = 2$, the normalized negative adjacency matrix $M = -D^{-1/2} A D^{-1/2} \otimes \begin{pmatrix} 1 & -1 \\ -1 & 1 \end{pmatrix}$ and $E = D \otimes \Id_2$. The speedup to near-linear time is described in \Cref{subsec:fast-sdp}.
\end{comment}

Note that the main bottleneck for runtime in~\cref{alg:fast-csp} is solving the SDP. The approximation guarantees of our algorithm will follow via the analysis in the proof of~\cref{th:main-body} and thus it suffices to show that we can find an (approximate) solution to the SDP in~\cref{alg:fast-csp} in near-linear time. We utilize existing techniques from~\cite{fastsdp2csps10}, which give a fast SDP solver that runs in time $\poly(k, q, 1/\eps) \cdot \widetilde{O}(n + m)$. Since our setup is very similar to theirs and our improved runtime largely follows from prior work, we will only sketch the main differences in the proof.

\begin{proof}[Proof of \cref{thm:near-linear-2csp-label-extended}]
    Note that the approximation guarantees of the theorem follow from~\cref{th:main-body}, provided that we can quickly find an approximate SDP solution where on average constraints are approximately satisfied\footnote{As noted in~\cite{fastsdp2csps10} if they are approximately satisfied on average, then by~\cite{10.1109/FOCS.2009.74} we can modify the solution to satisfy all constraints exactly at the expense of decreasing the objective value by $\eps q^2 m$.} and our additional constraint holds up to accuracy $\eps \Tr(D) / 100$. Below, we will briefly sketch the algorithm and analysis of the fast SDP solver.

    We will find an (approximate) solution to the relevant SDP in near-linear time via matrix multiplicative weights, as in~\cite{fastsdp2csps10}.
    Note that we can represent all the relevant moments of our pseudo expectation via the following matrix:
    \[ \pE \left( 1 \quad \frac{1}{\sqrt{qn}} y \right)^\top \left( 1 \quad \frac{1}{\sqrt{qn}} y \right) \,.\]
    Observe that the bottom right $qn \times qn$ submatrix is a PSD matrix with trace $1$ and corresponds directly to the SDP solution computed in prior work~\cite{fastsdp2csps10}. Note that via~\cite{fastsdp2csps10} we can find such a matrix which optimizes the objective function in near-linear time, without taking into account our additional constraint. The main remaining piece is to show that we can also enforce our new constraint in the matrix multiplicative weights solver. The two ingredients towards this will be to show that (a) the new constraint is low-width and (b) we can efficiently detect violations of the constraint only given access to a sketch of the proposed ``feasible'' solution.

    First, $\pE_{\mu} \norm{\Pi E^{1/2} x - v}^2 \leq \eps \Tr (D)$ is indeed a low-width constraint. The hyperplane corresponding to this constraint has operator norm $O(\Tr (E))$, while we need the constraint to be satisfied up to $\eps/100 \cdot \Tr (D)$ slack. Thus, the width of the constraint is $O(q/\eps)$.

    We now describe how to efficiently detect a violation of the constraint given a sketched version of the solution. Recall that every iteration we are given $\widetilde{W}_t \in \mathbb{R}^{\log (qn) \times (qn)}$ such that $\widetilde{W}^\top \widetilde{W}$ is the sketched SDP solution. Note that in order to detect the violation of the constraint in the sketched SDP, we only need to know $\Pi' (E')^{1/2} \widetilde{W}^\top \widetilde{W} (E')^{1/2} \Pi'$, where $\Pi'$ is the projection of the last $n$ coordinates to the relevant subspace and $E'$ is a $(n+1) \times (n+1)$ dimensional block matrix with $E$ in the lower right block, $\Id_1$ in the top left block, and zeros elsewhere\footnote{$\Pi'$ is not equal to $\Pi$ since our submatrix also includes an initial row and column for the expectation.} Note that this is a low rank matrix, since $\Pi'$ has rank $k+1$. Thus, if we just compute $B \Pi' \widetilde{W}_t$, where $B \in \mathbb{R}^{(k+1) \times n}$ simply is a change of basis, then by looking at $(B \Pi' \widetilde{W}_t)^\top (B \Pi' \widetilde{W}_t)$ we can determine if $\widetilde{W}^\top \widetilde{W}$ satisfies the constraint. The fact that it suffices to check whether the solution would be satisfied by the sketch follows from prior analysis.

    Furthermore, we can efficiently compute $(B \Pi' \widetilde{W}_t)^\top (B \Pi' \widetilde{W}_t)$. It suffices to be able to compute $B \Pi'$ in near-linear time, and this matrix can be constructed quickly given the relevant top eigenvectors of $M$. Furthermore, note that instead of having $\Pi$ be the projection to the eigenspaces with eigenvalues at least $\eps$, it also suffices to have $\Pi$ be the projection to some subspace which includes all eigenvectors with eigenvalues at least $2\eps$ and has dimension at most $\rank_{\geq \eps} (M)$.\footnote{We simply lose a factor of $2$ in the approximation guarantee, which still yields a $O(\eps q^2 m)$ additive approximation.} Thus, it suffices to find a collection of at most $\rank_{\geq \eps} (M)$ vectors such that every eigenvector with eigenvalue at least $2 \eps$ is in close to the span. We can compute this via the power method---note that between $2\eps$ and $\eps$ there must exist a gap between two eigenvectors, and thus it suffices to take $\log n$ iterations of a top-$k$ power method. Each iteration takes time $O(kq^2m)$ since $M$ has at most $q^2m$ non-zero entries.

    Finally, we will argue that we can implement our rounding algorithm when only given access to the sketched solution, without computing the full $qn \times qn$ solution matrix. In order to round the solution, it suffices to know $\pE y$, since this defines the marginal distribution on each coordinate, and we can then independently round each coordinate. However, this is a $1 \times qn$ submatrix of $\widetilde{W}^\top \widetilde{W}$ which can be reconstructed in $\widetilde{O}(qn)$ time from the Gram vectors.
\end{proof}

\maxcut*

Given the proof of~\cref{thm:near-linear-2csp-label-extended}, the proof of~\cref{th3:maxcut} follows via some small modifications. 

\begin{proof}[Proof of~\cref{th3:maxcut}]
    We aim to apply~\cref{thm:near-linear-2csp-label-extended}, but with a small twist. Note that we can write the objective value corresponding to the MAX-CUT in~\cref{th:main-body} as follows:
    \[ \frac{1}{2} \left[ y^\top M_{A}^1 y + y^\top \left( -A \otimes \begin{pmatrix} 1 & -1 \\ -1 & 1 \end{pmatrix} \right) y \right]\,,\]
    where $M_{A}^1$ is the label-extended graph of the CSP with constraint graph $A$ and every constraint being trivially satisfied. Since the first term has a fixed value for every $y \in \calC_2^n$, it is equivalent to optimizing the second term, or the quadratic optimization problem associated with $M = -D^{-1/2} A D^{-1/2} \otimes \begin{pmatrix} 1 & -1 \\ -1 & 1 \end{pmatrix}$ and $E = D \otimes \Id_2$. Thus, by~\cref{th:main-body} we have that~\cref{alg:fast-csp} incurs a small rounding error for $k = \rank_{\ge \eps} (M) = O\left(\rank_{\le -\eps} (A)\right)$ when utilizing the matrices (and associated projections to the top eigenspaces) described in the prior line. Furthermore, since the original formulation of the objective corresponds to a quadratic optimization associated with the label-extended graph of a CSP, we can apply~\cref{thm:near-linear-2csp-label-extended} to show that~\cref{alg:fast-csp} can be implemented in $\wt{O} \left( n + m \right) \cdot \left(\frac 1 \eps\right)^{O(k)} \cdot \poly(q)$ time since the relevant SDPs are equivalent (up the this $y^\top M_{A}^1 y$ fixed additive shift) to optimization on the label-extended graph of a CSP. 
\end{proof}

% \begin{mdframed}
%   \begin{algorithm}
%     \label{algo:max-cut}\mbox{}
%     \begin{description}
%     \item[Input:] $A \in [0,1]^{n \times n}$ and $0 < \eps < 1$.
    
%     \item[Operations:]\mbox{}
%     \begin{enumerate}
%         \item Let $\calS$ be a $\sqrt{\eps n}$-fineness net of $\sqrt{n} \cdot \mathbb{S}^k$ and let $\Pi$ be the projection to the span of $A_{[\eps, 1]}$.
%         \item For $v \in \calS$:
%         \begin{enumerate}
%             \item Let $\mu$ be a degree $2$ SoS psuedo-distribution over $\{\pm 1\}^n$ optimizing
%             \begin{equation}
%             \label{eqn:main-opt-max-cut}
%             \begin{split}
%                 & \min_{\pE_{\mu} }   \hspace{0.1in}\pE_{\mu} \langle x, Ax \rangle \\
%                 & \textrm{s.t.}  \hspace{0.1in} \pE_{\mu} \norm{\Pi x - v}^2 \leq \eps n
%             \end{split}
%             \end{equation}
%             \item Round $\mu$ to $x_v \in \{\pm 1\}^n$ by independently sampling each $x_i$ from the local distribution. 
%         \end{enumerate}
%         \item Calculate the size of the cut associated with each $x_v$ and output the $x_v$ achieving the largest cut value.
%     \end{enumerate}
    
%     \item[Output:] A cut $x \in \{\pm 1\}^n$.
%     \end{description}
%   \end{algorithm}
% \end{mdframed}

%% file: content/rank-bound.tex
%!TEX root = ./main.tex

\section{Threshold Rank of the Signed and Label-Extended Adjacency Matrix}
\label{sec:label-extended-threshold-rank}

We prove the following theorem relating the threshold rank of real-valued signings of adjacency matrices to the threshold rank of the underlying adjacency matrix.
The proof is a small adaptation of an argument from \cite{BHSV25}, which is itself an adaptation of an argument from \cite{BRS11}.

\begin{theorem}
    \label{th:threshold-rank-bound}
   Suppose $A \in \R^{n \times n}_{\geq 0}$ is symmetric with $\|A\| \leq 1$.
   Let $B \in \R^{n \times n}$ be symmetric with $|B_{ij}| \leq A_{ij}$ for all $i,j$.
   % And suppose that $\rank_{\geq \lambda}(B) \geq t$ for some $\lambda \geq 0$ and $t$ a positive integer.
   If $\rank_{\geq \tau}(A) \leq s$, then for any $\sigma > 0$,
   \[
   \rank_{\geq \sqrt{\tau (1 - \sigma) + \sigma}}(B) \leq \frac s {\sigma^2} \, .
   \]
\end{theorem}

\begin{corollary}
    \label{cor:label-extended-threshold-rank-bound}
    Let $A \in \R^{n \times n}_{\ge 0}$ with $\|A\| \le 1$.
    For some $q \in \N$, let $B \in \C^{nq \times nq}$ be such that for any $i,j \in [n]$ and $\alpha,\beta \in [q]$, $B_{(i,\alpha),(j,\beta)} = 0$ if $A_{ij} = 0$, and $|B_{(i,\alpha),(j,\beta)}| \le A_{ij}$. 
    % Suppose that $\rank_{\geq \lambda}(B) \geq t$ for some $\lambda \geq 0$.
    If $\rank_{\geq \tau}(A) \leq s$, then
    \[ \rank_{\geq q\sqrt{\tau(1-\sigma) + \sigma}}(B) \leq \frac{s}{\sigma^2} \mper \]
    In particular, for any $\eps > 0$, setting $\sigma = \tau = \eps^2$,
    \[ \rank_{\geq 2q\eps}(B) \le \frac{\rank_{\ge \eps^2}(A)}{\eps^4} \mper \]
\end{corollary}
\begin{proof}
    This immediately follows on instantiating \Cref{th:threshold-rank-bound} with $\frac{1}{q} A \otimes \bone_q \bone_q^\top$ and $\frac{1}{q} B$.
\end{proof}

Given this final ingredient, we may prove our result \Cref{th1:arbitrary-csps} on near-linear time algorithms for general 2CSPs on bounded threshold rank graphs, restated for convenience.

\arbitrarycsps*

\begin{proof}
    Let $B \in \R^{nq \times nq}$ be the normalized adjacency matrix of the label-extended graph: For $i,j \in [n]$ and $\alpha,\beta \in [q]$, $B_{(i,\alpha),(j,\beta)} = (D^{-1/2} A D^{-1/2})_{ij}$ if $\{i,j\}$ is an edge and $\varphi_{ij}(\alpha,\beta) = 1$, and is $0$ otherwise. Now instantiate \Cref{thm:near-linear-2csp-label-extended} with $M = \frac{1}{q} B$, the scaled degree diagonal $E = qD \otimes \Id_q$, and the error parameter as $2\eps$. To conclude, the $2\eps$-threshold rank of $\frac{1}{q} B$ can be bounded in terms of the $\eps^2$-threshold rank of $A$ using \Cref{cor:label-extended-threshold-rank-bound}.
\end{proof}

\cite{BHSV25} proves \Cref{th:threshold-rank-bound} in the case that $B$ is $-A$.
We show that with minor changes the same proof allows $B$ to have complex entries.
The only change needed to the argument which proves \cite[Corollary 4.2]{BHSV25}, is the following lemma which substitutes for \cite[Lemma 4.4]{BHSV25}.

\begin{lemma}[{Adapted from \cite[Lemma 4.4]{BHSV25}}]
    Suppose $A \in \R^{n \times n}_{\geq 0}$ is symmetric with $\|A\| \leq 1$.
    Let $B \in \R^{n \times n}$ be symmetric with $|B_{ij}| \leq A_{ij}$ for all $i,j$: in particular, $B_{ij} = 0$ if $A_{ij} = 0$.
    Suppose that $\rank_{\geq \lambda}(B) \geq t$ for some $\lambda \geq 0$ and $t$ a positive integer.
    Then there exists $V \in \R^{t^2 \times n}$ such that
    \[
    \iprod{A, V^\top V} \geq \lambda^2 \, , \quad \|V^\top V\|_F^2 \leq \frac 1 t \, , \quad \Tr(V^\top V) = 1 \, .
    \]
\end{lemma}

\begin{proof}[Proof, adapted from \cite{BHSV25}]
  Let $\{u_1,\ldots,u_t\}$ be $t$ orthonormal eigenvectors of $B$ whose corresponding eigenvalues are at least $\lambda$.
  Let $U \in \R^{n \times t}$ have $s$th column equal to $u_s / \sqrt{t}$.
  We record for later that
  \[
  \begin{gathered}
      \iprod{B, UU^\top} = \sum_{s \leq t} \frac 1 t \iprod{B, u_s u_s^\top} \geq \lambda \mcom \\
      \|UU^\top\|_F^2 = \|U^\top U\|_F^2 = \frac 1 t \mcom \\
      \Tr(UU^\top) = \Tr(U^\top U) = 1 \mper
  \end{gathered}
  \]
  We denote the $i$-th row of $U$ by $w_i$, and let $v_i = w_i^{\otimes 2} / \|w_i\|$ if $w_i \neq 0$ and $v_i = 0$ otherwise.
  Let $V$ be the matrix whose $i$-th column is $v_i$.
  We will show that $V$ satisfies the conclusions of the lemma.

  For the first conclusion, note that
  \begin{align*}
  \Paren{ \sum_{i,j} B_{ij} \iprod{w_i, w_j}}^2 & = \Paren{ \sum_{i,j} \frac{\sqrt{A_{ij}} \cdot \iprod{w_i, w_j}}{\sqrt{\|w_i\| \|w_j\|}} \cdot \frac{B_{ij}}{\sqrt{A_{ij}}} \sqrt{\|w_i\|\|w_j\|}}^2 \\
  & \leq \Paren{\sum_{ij} A_{ij} \frac{\iprod{w_i, w_j}^2}{\|w_i\|\|w_j\|}} \cdot \Paren{\sum_{i,j} \frac{|B_{ij}|^2}{A_{ij}} \|w_i\|\|w_j\|} \\
  & \leq \Paren{\sum_{ij} A_{ij} \frac{\iprod{w_i, w_j}^2}{\|w_i\|\|w_j\|}} \cdot \Paren{\sum_{i,j} A_{ij} \|w_i\|\|w_j\|} \\
  & \leq \Paren{\sum_{ij} A_{ij} \frac{\iprod{w_i, w_j}^2}{\|w_i\|\|w_j\|}} \Tr(UU^T) 
  \end{align*}
  where for the last line we used that $\|A\| \leq 1$.
  This gives the conclusion that $\iprod{A,VV^\top} \geq \lambda^2$.
  The remaining conclusions follow exactly the arguments in \cite{BHSV25}, Lemma 4.4.
\end{proof}